\documentclass[letterpaper, 10 pt, conference]{ieeeconf}  

\IEEEoverridecommandlockouts                              
\usepackage{graphics} 
\usepackage{epsfig} 
\usepackage{amsmath} 
\usepackage{amssymb}  
\usepackage{cite}

\makeatletter
\let\proof\@undefined                        
\let\endproof\@undefined                  
\makeatother
\usepackage{amsthm}
\usepackage{multirow}
\usepackage{bbm, dsfont}
\usepackage{bm} 
\usepackage{url}
\usepackage[]{algorithm2e}
\usepackage{balance}
\usepackage{color}
\usepackage{subfig}
\usepackage[font=singlespacing]{caption}
\graphicspath{{./Figures/}}
\usepackage{array}
\newcolumntype{L}[1]{>{\raggedright\let\newline\\\arraybackslash\hspace{0pt}}m{#1}}
\newcolumntype{C}[1]{>{\centering\let\newline\\\arraybackslash\hspace{0pt}}m{#1}}
\newcolumntype{R}[1]{>{\raggedleft\let\newline\\\arraybackslash\hspace{0pt}}m{#1}}

\usepackage{amsmath,amssymb}
\newtheorem{theorem}{Theorem}
\newtheorem{definition}{Definition}

\newtheorem{remark}{Remark}
\newtheorem{proposition}{Proposition}

\newtheorem{problem}{Problem}

\DeclareMathOperator*{\argmin}{argmin}
\DeclareMathOperator{\sgn}{sgn}

\title{\LARGE \bf
Tractable Compositions of Discrete-Time Control Barrier Functions 
with Application to Lane Keeping and Obstacle Avoidance}

\author{Matthew Cavorsi\;\;\; Mohammad Khajenejad\;\;\; Ruochen Niu\;\;\; Qiang Shen\;\;\; Sze Zheng Yong
\thanks{M. Cavorsi, M. Khajenejad, R. Niu and S.Z. Yong are with School for Engineering of Matter, Transport and Energy, 
Arizona State University, Tempe, AZ, USA; 
Q. Shen is with the School of Aeronautics and Astronautics, Shanghai Jiao Tong University, Shanghai, P.R. China
(email: {\tt \{mcavorsi,mkhajene,rniu6,szyong\}@asu.edu, qiangshen@sjtu.edu.cn}). }
\thanks{This work is partially supported by NSF grant CMMI-1925110.}
}

\begin{document}

\maketitle
\thispagestyle{empty}
\pagestyle{empty}
\begin{abstract}
This paper introduces control barrier functions for discrete-time systems, which can be shown to be necessary and sufficient for controlled invariance of a given set. Moreover, we propose nonlinear discrete-time control barrier functions for partially control affine systems that lead to controlled invariance conditions that are affine in the control input, leading to a tractable formulation that enables us to handle the safety optimal control problem for a broader range of applications with more complicated safety conditions than existing approaches. In addition, we develop mixed-integer formulations for basic and secondary Boolean compositions of multiple control barrier functions and further provide mixed-integer constraints for piecewise control barrier functions. Finally, we apply these discrete-time control barrier function tools to automotive safety problems of lane keeping and obstacle avoidance, which are shown to be effective in simulation.
\end{abstract}

\vspace{0.05cm}
\section{Introduction} \label{sec:intro}
\emph{Motivation.} 
Motivated by safety-critical applications such as adaptive cruise control systems \cite{ames2016control}, multi-agent systems \cite{wu2016safety} and footstep placement of bipedal robots \cite{nguyen2016dynamic,nguyen2016optimal}, 
several safety control approaches have been developed  to guarantee safety, in addition to addressing the stabilization problem. In these applications, it is of great interest and importance to ensure that the control algorithms are tractable and can be implemented at run-time. 

\emph{Literature review.} A variety of Lyapunov-like approaches have been developed to construct \emph{barrier certificates} and (controlled) invariant sets for ensuring system safety, 
both for autonomous systems, e.g., in \cite{prajna2005necessity,prajna2007framework,tee2009barrier} and for control systems, e.g., \cite{wieland2007constructive,aubin2009viability,aubin1990survey,aubin2011viability}.
Moreover, these Control
Barrier Functions (CBFs) can be combined with control Lyapunov functions, yielding Control Lyapunov Barrier Functions (CLBF), which 
have been shown in recent years to be a promising approach for jointly guaranteeing 
safety and stability.

Although CBFs and CLBFs have been extensively studied in the control and verification literature for a broad range of continuous-time systems 
\cite{ames2016control,romdlony2016stabilization,glotfelter2019hybrid,santoyo2019barrier,glotfelter2017nonsmooth,borrmann2015control,nguyen2016exponential,lindemann2018control,lindemann2019decentralized,djaballah2017construction,glotfelter2018boolean,wieland2007constructive} for applications such as model predictive control, obstacle/collision avoidance, eventuality properties  
or safety establishment and multiobjective control, there are only relatively few studies that address the design of CBF-based approaches for discrete-time dynamical systems. The work in \cite{agrawal2017discrete} extends the continuous-time CBF-based developed tools for safety-critical applications to discrete-time systems, and established that the extension is not straightforward because the resulting optimization problem is not necessarily convex and hence \emph{tractability} remains an unsolved issue, except for some special cases such as linear/linearized settings.

On the other hand, the authors in \cite{wills2004barrier} developed a barrier function based model predictive control for a class of nonlinear discrete-time dynamics, which 
hinges extensively on the stabilizability of the linearized system, while 
\cite{ahmadi2019safe} applied discrete-time barrier functions 
to derive necessary and sufficient conditions that ensure safety of a given set, for multi-agent partially observable Markov decision processes and further 
proposed conditions for checking Boolean compositions of barrier functions to represent more complicated safety sets. However, 
the assumption of finite and countable actions (i.e., control inputs) as well as the Markov assumption are essential for obtaining a tractable solution in \cite{ahmadi2019safe}. 

\emph{Contribution.}
In this paper, we present an arguably more straightforward formulation of control barrier functions for discrete-time systems than the formulations in \cite{agrawal2017discrete,ahmadi2019safe}, and show that it is the least restrictive in terms of the set of allowable safe inputs, which in turn guarantees optimality when combined with an optimal controller.  
Further, we propose a more general class of nonlinear control barrier functions for partially control affine systems that lead to invariance conditions that are affine in the control input and hence, resulting in tractable optimization problems. This enables us to handle the safety optimal control problem for a broader range of applications than the case with linear systems and linear CBFs considered in \cite{agrawal2017discrete}.


Moreover, we derive mixed-integer formulations for basic \emph{Boolean compositions} of multiple CBFs, as an alternative to \cite{ahmadi2019safe,glotfelter2018boolean}, and further provide mixed-integer constraints corresponding to secondary Boolean compositions (i.e., \emph{implies}, \emph{exclusive or} and \emph{equivalence}), as well as \emph{if-then-else} statements. These compositions, when combined with the tractable nonlinear CBFs, enable us to guarantee safety for more complicated non-convex or piecewise safe sets and for more general switched systems with non-smooth dynamics using tractable mixed-integer linear programs.

Then, equipped by these 
discrete-time CBF tools, we consider the discrete-time lane keeping problem for autonomous driving that was previously only achieved using a continuous-time formulation \cite{ames2016control}. In addition, we extend this to the obstacle avoidance problem where a vehicle can avoid an obstacle by choosing to either go around its left or right.


\section{Preliminaries and Problem Formulation} \label{sec:prelim2}

\subsection{Notations and Definitions}
$\mathbb{R}^n$ denotes the set of n-dimensional real numbers,  $\mathbb{N}$ is the set of natural numbers, and the set of positive integers up to $n$ is denoted by $\mathbb{Z}_n^{+}$. 
$\mathbf{0}_{m\times n}$ represents the matrix of zeros of appropriate dimensions. We will also make use of  the following definitions. 
\begin{definition}[SOS-1 Constraint \cite{gurobi}] \label{def:SOS1}
	A special ordered set of degree 1 (SOS-1) constraint\footnote{Off-the-shelf solvers such as Gurobi and CPLEX \cite{gurobi,cplex} can readily handle these constraints, which can significantly reduce the search space for integer variables in branch and bound algorithms.} is a set of integer, continuous or mixed-integer scalar variables for which at most one variable in the set may take a value other than zero, denoted as SOS-1: $\{v_1,\hdots,v_N\}$. For instance, if $v_i \neq 0$, then this constraint imposes that $v_j=0$ for all $j \neq i$.$v_1=\hdots=v_{i-1}=v_{i+1}=\hdots=v_N=0$.
\end{definition}

\begin{definition}[Partition] \label{def:partition}
A partition of a set/domain $\mathcal{P}$ is a collection of $| \mathcal{J}|$ disjoint subsets $\mathcal{P}_{j}$ 
such that $\bigcup \limits_{ j\in \mathcal{J}} {\mathcal{P}}_{j} = \mathcal{P}$. 
\end{definition}



\vspace{-0.5cm}
\subsection{Problem Statement} \label{sec:form2}
Consider the following class of discrete-time partially control affine systems
\begin{align} \label{eq:DT_sys}
\begin{array}{rl}
x_{k+1}\triangleq \begin{bmatrix}x_{1,k+1} \\ x_{2,k+1} \end{bmatrix}&=\begin{bmatrix}f_1(x_k) \\ f_2(x_k) \end{bmatrix}+\begin{bmatrix}\mathbf{0}_{n_1 \times m} \\ g(x_k) \end{bmatrix}u_k,\\
&\triangleq f(x_k) + \tilde{g}(x_k) u_k,
\end{array}
\end{align}
where at time $k \in \mathbb{N}$, $x_k \in  \mathbb{R}^{n}$ and $u_k \in U \subseteq \mathbb{R}^m$ are the state and control input vectors, respectively. We assume that the state vector can be partitioned into two parts as $x_{k}=\begin{bmatrix}x^\top_{1,k} & x^\top_{2,k} \end{bmatrix}^\top$, where the dynamics of $x_{1} \in \mathbb{R}^{n_1}$ ($0\le n_1 \le n$) is autonomously governed by the known 
vector field $f_1(.):\mathbb{R}^n \to \mathbb{R}^{n_1}$, i.e., it is not affected by the control input signal $u_k$, while the dynamics of $x_{2} \in \mathbb{R}^{n_2}$ ($0\le n_2 \le n$, $n_1+n_2=n$) is governed by the known vector field $f_2(.):\mathbb{R}^n \to \mathbb{R}^{n_2}$ and is affinely affected by the control input $u_k$ via the dynamics of $x_{2} \in \mathbb{R}^{n_2}$ through the function $g:\mathbb{R}^n \to \mathbb{R}^{n_2 \times m}$ 
as described in \eqref{eq:DT_sys}. 

Note that the system \eqref{eq:DT_sys} is a generalization of affine control systems with an additional structure that is common for many practical systems with higher order dynamics, e.g., mechanical systems with inertia. This structure will also help us to derive nonlinear control barrier functions that lead to tractable  constraints that are affine in the control input.



Specifically, this paper seeks to address two problems:

\begin{problem} [Synthesis of Tractable Nonlinear Control Barrier Functions]
\label{prob:PACBF}
    For the discrete-time system in the form of (\ref{eq:DT_sys}), synthesize a function such that the safety set is forward controlled invariant, where the invariance condition is affine in the control input (hence, leads to tractable constraints). 
\end{problem}

\begin{problem}[Compositions of Control Barrier Functions]\label{prob:composite}
Given multiple control barrier functions, find mixed-integer encodings of their (basic and secondary) Boolean compositions. Moreover, compose piecewise functions as mixed-integer conditions.
\end{problem}
 

The motivation behind Problem \ref{prob:PACBF} is to obtain nonlinear discrete-time control barrier functions  with tractable invariance constraints (henceforth called tractable DT-CBFs)  for run-time implementation, while their compositions in Problem \ref{prob:composite} enable us to handle more complex dynamics and safety conditions, including switched dynamics and non-smooth control barrier functions. In fact, the need for the latter capability is motivated by discrete-time automotive safety applications, and in particular, for lane keeping and obstacle avoidance, which we will present in Section \ref{sec:example}.

\section{Main Results} \label{sec:method}

This section addresses Problems \ref{prob:PACBF}, and \ref{prob:composite} and in the process, develops tools that enable  optimal safety control for autonomous driving in Section \ref{sec:example}.

\subsection{Tractable Discrete-Time Control Barrier Functions}

This subsection considers the problem of synthesizing tractable discrete-time control barrier functions (DT-CBF), as stated in Problem \ref{prob:PACBF}. First, we introduce an arguably more straightforward formulation than existing formulations in the literature and then, we propose a class of nonlinear DT-CBF for partially control affine systems \eqref{eq:DT_sys} that leads to tractable constraints in optimal control problems.\\[-0.25cm]


\subsubsection{Discrete-Time Control Barrier Functions}
Consider a (safe) set $\mathcal{S}$ defined as
\begin{equation}
\begin{aligned}
    \mathcal{S} &\triangleq& \{ x \in \mathbb{R}^n : h(x) \geq 0 \}, 
\end{aligned}
\label{eq:inv}
\end{equation}
where $h:\mathbb{R}^n \to \mathbb{R}$ is any well-defined function, including discontinuous and non-smooth functions, and $\partial \mathcal{S} \triangleq {\{x\in \mathbb{R}^n : }h(x)=0 {\}}$ defines the boundary of the set. 

Next, we present the notion of (forward) controlled invariance of a set $\mathcal{S}$ and the definition of a DT-CBF, and show that the existence of the DT-CBF is  both sufficient and necessary for controlled invariance. 

\begin{definition}
A set $\mathcal{S}$ is called \emph{(forward) controlled invariant} with respect to the system dynamics \eqref{eq:DT_sys} if for every initial state $x_0 \in \mathcal{S}$, there exists a control input $u_k \in \mathbb{R}^m$ such that state trajectory remains in $\mathcal{S}$ at all times, i.e., $x_k \in \mathcal{S}$, $\forall k \in \mathbb{Z}$.
\end{definition}

\begin{definition}[Discrete-Time Control Barrier Function]\label{def:DT-CBF}
For the discrete-time system (\ref{eq:DT_sys}), the function $h:\mathbb{R}^n \to \mathbb{R}$ is a discrete-time control barrier function (DT-CBF) for the (safe) set $\mathcal{S}$ as defined in \eqref{eq:inv}, if 
\begin{gather} \label{eq:DT-CBF}
    \sup_{u \in U} h(f(x) + \tilde{g}(x) u)\ge 0, \qquad \forall x \in \mathcal{S}.
\end{gather}
Moreover, for any $x \in \mathcal{S}$, we define the corresponding (safe) input set
\begin{gather} \label{eq:safeinput}
    K_{\mathcal{S}}(x)=\{u\in U: h(f(x) + \tilde{g}(x) u)\ge 0\}. 
\end{gather}
\end{definition}

\begin{theorem}
Consider the discrete-time system in \eqref{eq:DT_sys} and the (safe) set $\mathcal{S}$ as defined in \eqref{eq:inv}. Then, $\mathcal{S}$ is (forward) controlled invariant if and only if there exists a DT-CBF as described in Definition \ref{def:DT-CBF}.
\end{theorem}
\begin{proof} 
With $x_k=x$ and $u_k=u$ for any $x\in \mathcal{S}$ and $u \in K_\mathcal{S}(x)$ at any time step $k$, the inequality in \eqref{eq:DT-CBF} is satisfied by definition, which means that from \eqref{eq:DT_sys}, we have the following: 
\begin{gather}
    \sup_{u_k \in U} h(f(x_k)+\tilde{g}(x_k)u_k)=\sup_{u_k \in U} h(x_{k+1})\ge 0. \label{eq:h_k1}
\end{gather}
In other words, $x_k \in \mathcal{S}$ implies that $x_{k+1} \in \mathcal{S}$ with $u_k \in K_\mathcal{S}(x_k)$. Further, with the base case of $x_0 \in \mathcal{S}$ (by assumption), we have an inductive proof of sufficiency of the DT-CBF for controlled invariance of $\mathcal{S}$. The necessity can be shown by contraposition. Suppose \eqref{eq:DT-CBF} does not hold, then all $u_k \in U$ for some $x_k$ lead to $h(x_{k+1})< 0$, which means that $\mathcal{S}$ is not controlled invariant.
\end{proof}

Note that our DT-CBF definition is slightly different from the ones proposed in \cite{agrawal2017discrete,ahmadi2019safe}, which have additional terms involving $h(x_k)$ when compared with \eqref{eq:h_k1}. We believe that our definition is more intuitive and straightforward since it directly imposes the controlled invariance condition without any modifications. More importantly, we can show that the (safe) input set $K_\mathcal{S}(x)$ in Definition \ref{def:DT-CBF} is a (non-strict) superset of the corresponding input sets based on the definitions in \cite{agrawal2017discrete,ahmadi2019safe}, as shown in the following proposition.

\begin{proposition}
The (safe) input set $K_\mathcal{S}(x)$ for any $x\in \mathcal{S}$ corresponding to the DT-CBF in Definition \ref{def:DT-CBF} satisfies
$$K_\mathcal{S}(x)\supseteq K'_\mathcal{S}(x) \ \text{and} \ K_\mathcal{S}(x)\supseteq K''_\mathcal{S}(x),$$
where the input sets $K'_\mathcal{S}(x)$ and $K''_\mathcal{S}(x)$ defined by
\begin{align*}
    K'_\mathcal{S}(x) &=\{u\in U: h(f(x) + \tilde{g}(x) u)+(\gamma-1)h(x)\ge 0 \},\\
    K''_\mathcal{S}(x) &=\{u\in U: h(f(x) + \tilde{g}(x) u)\hspace{-0.05cm}+\hspace{-0.05cm}\alpha(h(x))\hspace{-0.05cm}-\hspace{-0.05cm}h(x)\ge 0 \},
\end{align*}
correspond to the definitions of DT-CBF in \cite[Proposition 4]{agrawal2017discrete} and \cite[Definition 2]{ahmadi2019safe}, respectively, with $0\le \gamma \le 1$ and a class $\mathcal{K}$ function, $\alpha \in \mathcal{K}$, satisfying $\alpha(h(x))<h(x)$.
\end{proposition}
\begin{proof}
The result follows directly from the observation that 
\begin{gather*}
   u \in K'_{\mathcal{S}}(x)\Rightarrow h(f(x) + \tilde{g}(x) u)\ge (1-\gamma)h(x)\ge 0 ,\\
    u \in K''_{\mathcal{S}}(x)\Rightarrow h(f(x) + \tilde{g}(x) u)\ge h(x)-\alpha(h(x))\ge 0 ,
\end{gather*}
for all $x \in \mathcal{S}$, with the above choices of $\gamma$ and $\alpha$, as well as $h(x)\ge 0$; hence, $u \in K_\mathcal{S}$ by definition in \eqref{eq:safeinput}.
\end{proof}

This means that the DT-CBF definitions in \cite{agrawal2017discrete,ahmadi2019safe} are sufficient for controlled invariance but only necessary with the choice of $\gamma=1$ and $\alpha(h(x))=h(x)$. Further, the (safe) input set is the least restrictive when using the DT-CBF in Definition \ref{def:DT-CBF} and when incorporated into an optimal safety controller, does not lead to suboptimality. To our understanding, the extra terms in \cite{agrawal2017discrete,ahmadi2019safe} are a legacy from their continuous-time predecessors, e.g., \cite[Definition 5]{ames2016control}, where a relaxation of the invariance condition is introduced to extend the condition for only the boundary of the set $\mathcal{S}$ to the entire domain, including its interior. 
However, this is not needed for the discrete-time counterpart because the controlled invariance condition in \eqref{eq:DT-CBF} is already a necessary and sufficient condition for the entire  set $\mathcal{S}$. Nevertheless, the extra terms in the previous definitions may still be helpful when there are small modeling uncertainties.\\[-0.25cm]

\subsubsection{Tractable DT-CBF for Partially Control Affine Systems}

An important consideration when deriving a control barrier function is the tractability of the resulting controlled invariance condition in \eqref{eq:DT-CBF}. As observed in \cite{agrawal2017discrete}, unlike the continuous-time counterpart, the invariance condition when incorporated as a constraint in an optimal control problem will in general lead to nonlinear constraints and hence, the authors in \cite{agrawal2017discrete} focused only on linear systems with linear DT-CBFs. Indeed, this special case is the only one where the controlled invariance condition in \eqref{eq:DT-CBF} is affine in the control input for control affine systems in \eqref{eq:DT_sys} with $n_1=0$. 

However, when additional structure is present, i.e., when $n_1>0$ for systems with higher order dynamics, this class of partially control affine systems can also lead to controlled invariance conditions in \eqref{eq:DT-CBF} that are control affine with a careful choice of DT-CBFs, which we introduce next.

\begin{definition}[Partially Control Affine DT-CBF] \label{def:PACBF}
 For a discrete-time partially control affine system in the form of (\ref{eq:DT_sys}), the function $h_A:\mathbb{R}^n \to \mathbb{R}$ satisfying
 \begin{align} \label{eq:PACBF}
 h_A(x) = \mu(x_{1}) x_{2} + \nu(x_{1})  
 \end{align}
 is a discrete-time partially control affine control barrier function for the (safe) set $\mathcal{S}$ as defined in \eqref{eq:inv}, if
\begin{align} 
   \displaystyle \hspace{-0.2cm} &\sup_{u \in U} h_A(f(x) + \tilde{g}(x) u) \label{eq:PCA}\\
   \nonumber \displaystyle \hspace{-0.2cm} &=\sup_{u \in U} \mu(f_1(x))(f_2(x)+g(x)u) + \nu(f_1(x)) \ge 0, \forall x \in \mathcal{S},\hspace{-0.4cm}
\end{align}
where $\mu:\mathbb{R}^{n_1} \to \mathbb{R}^{1 \times n_2}$ and 
$\nu:\mathbb{R}^{n_1} \to \mathbb{R}$ 
are any nonlinear functions. 
Moreover, for any $x \in \mathcal{S}$, we define the corresponding (safe) affine input set
\begin{gather*}
    K^A_{\mathcal{S}}(x) \hspace{-0.05cm}=\hspace{-0.05cm}\{u\in U\hspace{-0.05cm}:\hspace{-0.05cm} \mu(f_1(x))(f_2(x)\hspace{-0.05cm}+\hspace{-0.05cm}g(x)u) \hspace{-0.05cm}+\hspace{-0.05cm} \nu(f_1(x))\hspace{-0.05cm}\ge \hspace{-0.05cm} 0 \}.
\end{gather*}
\end{definition}

\begin{remark}
The controlled invariance condition in \eqref{eq:PCA} is affine in the control input, as desired. Thus, when included as a tractable constraint in an optimal control problem with a quadratic cost, the result is a quadratic program (QP), similar to the continuous-time safety control approach in \cite{ames2016control}.
\end{remark}


\subsection{Compositions of Multiple and Piecewise DT-CBFs}\label{sec:compose}
Next, we develop tools for encoding Boolean compositions of multiple discrete-time control barrier functions as well as piecewise/non-smooth control barrier functions as mixed-integer constraints. 

First, we analyze three basic Boolean operations for composition of multiple control barrier functions $\{h_i(x)\}_{i\in \mathbb{Z}^+_N}$, i.e., $\neg$ (negation), $\land$ (conjunction) and $\vee$ (disjunction). 
The negation operator is trivial and can be shown by checking if $-h_i(x)$ satisfies the invariance property. Formally, we have
\begin{align}\label{eq:negation}
   \neg  h_i(x) \ge 0 \equiv 
   h_i(x) < 0 \Big. 
\end{align}
For the disjunction operator $\vee$, we can represent them as 

\vspace{-0.35cm}\begin{align} \label{disj_operator}
\displaystyle\bigvee_{i=1}^{N} h_i(x)\ge 0 \equiv 
\begin{array}{c}
\Big \{  
\forall i \in \mathbb{Z}_N^+: h_i(x) \geq s_i, \text{SOS-1}:\{ s_i, b_i\}, \\ 
b_i \in \{ 0,1\}, \sum_{i=1}^N b_i \ge 1 \Big\},\\[-0.6cm]
\end{array}
\end{align}

\noindent with $s_i$ being a slack variable, which ensures that there exists at least one $j \in \mathbb{Z}_N^+$ such that $h_j(x) \geq 0$. Moreover, for the conjunction operator $\land$, we have
\begin{align} \label{eq:conj}
\bigwedge_{i=1}^{N} h_i(x)\ge 0 \equiv \Big \{ 
\forall i \in \mathbb{Z}_N^+: h_i(x) \geq 0
\Big\},
\end{align}
which enforces that $h_j(x) \geq 0$ for all $j \in \mathbb{Z}_N^+$.

By leveraging the above three basic Boolean operations, we can further compose the following three secondary Boolean operations found in Boolean algebra: 
\begin{align}
  \hspace{-0.1cm}  & h_i(x) \rightarrow h_j(x)  \triangleq \neg h_i(x) \vee h_j(x) ,\label{eq:implies}\\
  \hspace{-0.1cm}  & h_i(x) \oplus h_j(x)  \triangleq (h_i(x) \vee  h_j(x)) \wedge \neg ( h_i(x) \wedge h_j(x)) , \hspace{-0.15cm} \\
    \hspace{-0.1cm} & h_i(x) \equiv h_j(x) \triangleq \neg (h_i(x) \oplus h_j(x)) ,
\end{align}
which represent the \textit{implication}, \textit{exclusive or} and \textit{equivalence} operations of a pair of control barrier functions  $h_i(x)$ and $h_j(x)$, respectively, where we suppressed the $\ge 0$ terms in the above for the sake of brevity and clarity.

Finally, we consider the composition of piecewise control barrier functions that enable us to represent more complicated non-convex safe sets, e.g., for the lane keeping problem in Section \ref{sec:example}. 
Specifically, for a partition of the domain $\bigcup \limits_{ j\in \mathcal{J}} {\mathcal{P}}_{j}$ (cf. Definition \ref{def:partition}), where each subregion is represented by the inequality $p_j(x)\le 0$, the partition/mode-dependent control barrier function can be expressed by a \emph{if-else} statement in the form of
\begin{align}
h_j(x_k) \ge 0 \text{ if } p_j(x)  \le 0,
\end{align}
that can be written using the \emph{implication} operator as 
\begin{align} \label{PLBF}
   p_j(x) \le 0 \rightarrow h_j(x) \Leftrightarrow \neg(p_j(x) \le 0) \vee h_j(x).
\end{align}
Then, with the negation and disjunction operators defined in \eqref{eq:negation} and \eqref{disj_operator}, we can encode \eqref{PLBF} as mixed-integer constraints.


Similar to the discussion above on tractability of the controlled invariance condition when added as a constraint in an optimal control problem, we will define a piecewise DT-CBF for partially control affine systems in \eqref{eq:DT_sys} lead to mixed-integer linear constraints, as follows:


\begin{definition}[Piecewise Partially Control Affine DT-CBF] \label{def:PPACBF}
 For a discrete-time piecewise partially control affine system in the form of (\ref{eq:DT_sys}), the piecewise function $h_P:\mathbb{R}^n \to \mathbb{R}$ satisfying
 \begin{align} \label{eq:PPACBF}
 h_{P}(x) \hspace{-0.05cm}=\hspace{-0.05cm} \{\mu_j(x_{1}) x_{2} \hspace{-0.05cm}+\hspace{-0.05cm} \nu_j(x_{1}) \ \text{if} \ \kappa_j(x_{1}) x_{2} \hspace{-0.05cm}+\hspace{-0.05cm} \lambda_j(x_{1}) \hspace{-0.05cm} \le \hspace{-0.05cm} 0\}_{j=1}^{|\mathcal{J}|}
 \end{align}
 is a discrete-time partially control affine control barrier function for the (safe) set $\mathcal{S}=\bigcup \limits_{ j\in \mathcal{J}}\{ {\kappa_j(x_{1}) x_{2} + \lambda_j(x_{1})\le 0}\}_{j}$ as defined in \eqref{eq:inv}, if
\begin{align} 
   \displaystyle \sup_{u \in U} h_{P}(f(x) + \tilde{g}(x) u) \ge 0, \quad \forall x \in \mathcal{S},  
\end{align}
and equivalently, for all $j \in \mathbb{Z}^+_{|\mathcal{J}|}$,
\begin{align}\label{eq:PDT-CBF}
\begin{array}{r}
    \sup\limits_{u \in U} \mu_j(f_1(x))(f_2(x)+g(x)u) + \nu_j(f_1(x)) \ge 0,  \\
    \text{if} \  \kappa_j(f_1(x))(f_2(x)+g(x)u) + \lambda_j(f_1(x))\le 0,
    \end{array}
\end{align}
where $\mu_j,\kappa_j:\mathbb{R}^{n_1} \to \mathbb{R}^{1 \times n_2}$ and 
$\nu_j,\lambda_j:\mathbb{R}^{n_1} \to \mathbb{R}$ 
are any nonlinear functions. 
Moreover, we define the corresponding (safe) piecewise affine input set
\begin{gather*}
\begin{array}{r}
    K^{P}_{\mathcal{S}}(x) \hspace{-0.1cm}=\hspace{-0.1cm} \{u\in U\hspace{-0.05cm}:\hspace{-0.05cm} \mu_j(f_1(x))(f_2(x)\hspace{-0.1cm}+\hspace{-0.1cm}g(x)u) \hspace{-0.1cm}+\hspace{-0.1cm} \nu_j(f_1(x))\hspace{-0.1cm}\ge \hspace{-0.1cm} 0\\
    \hspace{0.45cm} \text{if} \ \kappa_j(f_1(x))(f_2(x)\hspace{-0.05cm}+\hspace{-0.05cm}g(x)u) \hspace{-0.05cm}+\hspace{-0.05cm} \lambda_j(f_1(x))\hspace{-0.05cm}\le \hspace{-0.05cm} 0, \ j \in \mathcal{J} \}.
\end{array}
\end{gather*}
\end{definition}

\begin{remark}\label{rem:MIL}
The controlled invariance condition in \eqref{eq:PDT-CBF} is piecewise control affine and hence, when incorporated as a  constraint in an optimal control formulation with a quadratic cost, the result is a mixed-integer quadratic program (MIQP). Similar results can also be derived in a straightforward manner when the system dynamics are switched among a set of partially control affine dynamics, and thus, a detailed description is omitted for the sake of brevity.
\end{remark}



\section{Application to Lane Keeping and Obstacle Avoidance} \label{sec:example}

In this section, we apply the proposed DT-CBF tools to two automotive safety applications, namely Lane Keeping (LK) and Obstacle Avoidance (OA).

The goal of the Lane Keeping (LK) problem is to keep a vehicle in the middle of a desired lane that may be curved by controlling the vehicle's lateral displacement. The simulation example conveyed in this work was largely inspired by the LK example in \cite{ames2016control},
where the authors developed a continuous-time CBF based approach to solve this problem. By contrast, we consider the development of a discrete-time CBF approach and show that the resulting optimal control problem is ``harder" in that we now have a mixed-integer quadratic program (MIQP) as opposed to a quadratic program (QP) in \cite{ames2016control}. Nonetheless, we believe that this discrete-time implementation is important since almost all current controllers on smart and autonomous systems, including vehicles, are digital. 

Next, the LK capability is extended to allow the vehicle to avoid an obstacle in the road lane by using the compositions described in Section \ref{sec:compose} to choose to either drive around the obstacle to the left, or to the right. 

\subsection{Lane Keeping Setup}


Similar to \cite{ames2016control}, we consider a time-discretized version of the vehicle model 
in \cite{rossetter2006lyapunov} (using the forward Euler method with sampling time $t_s$):
\begin{equation}
    x'_{k+1} = (I+At_s) x'_k + B t_s u_k + E t_s r_{d,k} ,
    \label{eq:model}
\end{equation}
where
\begin{equation*}
    \footnotesize{A \hspace{-0.05cm}=\hspace{-0.05cm} \begin{bmatrix} 0 & 1 & V_0 & 0 \\ 0 & -\frac{C_f + C_r}{M V_0} & 0 & \frac{b C_r - a C_f}{M V_0} \hspace{-2.5pt} - \hspace{-2.5pt} V_0 \\ 0 & 0 & 0 & 1 \\ 0 & \frac{b C_r - a C_f}{I_z V_0} & 0 & -\frac{a^2C_f + b^2C_r}{I_z V_0} \end{bmatrix}\hspace{-0.1cm}, B \hspace{-0.05cm}=\hspace{-0.05cm} \begin{bmatrix} 0 \\ \frac{C_f}{M} \\ 0 \\ a\frac{C_f}{I_z} \end{bmatrix}\hspace{-0.1cm}, E \hspace{-0.05cm}=\hspace{-0.05cm} \begin{bmatrix} 0 \\ 0 \\ -1 \\ 0 \end{bmatrix}}\hspace{-0.05cm}.
\end{equation*}

The states $x'_k = \begin{bmatrix} y_k & \nu_k & \psi_k & r_k \end{bmatrix}^\top$ are the lateral displacement of the car from the center of the lane ($y_k$), the car's lateral velocity ($\nu_k$), the yaw angle of the car with respect to the lane center ($\psi_k$), and the yaw rate of the car ($r_k$). The input $u_k$ is the angle of the front tires at the current time step $k$. Road curvature is modeled as a \emph{known} disturbance to the system, 
and the road curves at a rate of $r_{d,k} = \frac{V_0}{R_k}$ where $V_0$ is the longitudinal velocity of the vehicle and $R_k$ is the radius of curvature of the road at time step $k$. The parameters $M$ and $I_z$ are the vehicle mass and moment of inertia about the center of mass, respectively, $a$ and $b$ are the distance from the center of mass to the front and rear tires, respectively, and $C_f$ and $C_r$ are tire parameters.

First, we put the system in \eqref{eq:model} into the partially control affine form in \eqref{eq:DT_sys} with a reduced state $x_k=\begin{bmatrix} y_k & \nu_k \end{bmatrix}^\top$ with $\psi_k$ and $r_k$ as known/measured parameters, and with
\begin{align*}
\begin{array}{c}
    g(x_k)=\frac{C_f}{M},\ f_1(x_k)=\begin{bmatrix} 1 & t_s \end{bmatrix}x_k +t_s V_0 \psi_k,\\
    f_2(x_k)=\begin{bmatrix} 1 & -t_s\frac{C_f + C_r}{M V_0}  \end{bmatrix}x_k + t_s \frac{b C_r - a C_f}{M V_0} r_k.
    \end{array}
\end{align*}

We consider two constraints 
in the LK problem: 

\subsubsection{\textbf{Acceleration Constraint}} The first constraint is to prevent unbounded lateral acceleration $a_k$ of the car:
\begin{equation} \label{eq:acc_c}
    |a_k|=\left| \frac{v_{k+1} - v_k}{t_s} \right| \leq a_{max}, 
     \ \forall k \in \mathbb{N},
\end{equation}
where $v_k = \frac{y_{k+1}-y_k}{t_s}$ is the instantaneous lateral velocity. 
From \eqref{eq:model}, 
$\frac{v_{k+1}-v_k}{t_s}= \frac{y_{k+2}-2y_{k+1}+y_k}{t_s^2} $ can be simplified to
\begin{align*}
\begin{array}{rl}
    \frac{v_{k+1}-v_k}{t_s}&= -\frac{C_f + C_r}{MV_0} \nu_k + \frac{bC_r - aC_f}{MV_0}r_k - V_0r_{d,k} + \frac{C_f}{M}u_k \\
    &= -\frac{1}{M}F_0 + \frac{C_f}{M} u_k ,
        \end{array}
\end{align*}
where
    $F_0 = C_f \frac{\nu_k + a r_k}{V_0} + C_r \frac{\nu_k - b r_k}{V_0} + M V_0 r_{d,k}$.
The constraint \eqref{eq:acc_c} can then be written as
\begin{equation}
    \begin{bmatrix} 1 \\ -1 \end{bmatrix} u_k \leq \begin{bmatrix} \frac{1}{C_f} \left( M a_{max} + F_0 \right) \\ \frac{1}{C_f} \left( M a_{max} - F_0 \right) \end{bmatrix} . \tag{AC}
    \label{eq:AC}
\end{equation}

\subsubsection{\textbf{Lane Centering Constraint}} This second constraint keeps the car from drifting too far away from the middle of the lane, where it could possibly drift out of it. This can be done by restricting the maximum lateral displacement:
\begin{equation}
    |y_k| \leq y_{max}, \ \forall k \in \mathbb{N}. 
    \label{eq:LCC}
\end{equation}
As described in \cite{ames2016control}, a typical United States lane is 12 feet wide while a car is about 6 feet wide, so the maximum lateral displacement the car can safely experience is 3 feet to either side, so $y_{max} = 3 \text{ feet} \approx 0.9 \text{ meters}$.

The next proposition proposes a DT-CBF that can enforce the controlled invariance of the lane centering constraint as a safe set, i.e., $\mathcal{S}_{LK}=\{x\in \mathbb{R}^2: \eqref{eq:LCC} \ \text{holds}\}$, subject to the acceleration input constraint, i.e., $U=\{u \in \mathbb{R}:\eqref{eq:acc_c} \ \text{holds}\}$. 

\begin{proposition}
The function $h_{LK}:\mathbb{R}^2\to \mathbb{R}$ 
\begin{equation}
\begin{array}{rl}
    h_{LK}(x) = &\sqrt{2a_{max}(y_{max}-\sgn(v)y)+\frac{1}{4}a_{max}^2 t_s^2}\\ &- |v| - \frac{1}{2}a_{max}t_s,
\end{array} \label{eq:h_lk}
\end{equation}
where $v=\nu+V_0 \psi$ is the instantaneous lateral velocity, 
is a valid DT-CBF for the (safe) set $\mathcal{S}_{LK}=\{x\in \mathbb{R}^2: \eqref{eq:LCC} \ \text{holds}\}$. Moreover, 
the corresponding (safe) input set $K_{\mathcal{S}_{LK}}(x)$ (cf. Definition \ref{def:PPACBF}) for any $x \in \mathcal{S}$ can be implemented as mixed-integer linear constraints.
\end{proposition}

\begin{proof} First, we construct the safe set $\mathcal{S}$ by showing that $h_{LK}(x)\ge 0$ is equivalent to \eqref{eq:LCC}. 
For any (initial) displacement $y$ and instantaneous velocity $v$, with the maximum allowable acceleration/deceleration given $a = -\sgn(v)a_{max}$ (cf. \eqref{eq:acc_c})
it takes time $T = \frac{|v|}{a_{max}t_s}$ to reach $v_T = 0$. Correspondingly, the furthest lateral displacement with maximum acceleration/deceleration to come to a full stop is given by
\begin{align*}
\begin{array}{rl}
    y_T &= y + v t_s T - \frac{1}{2} \sgn(v)t_s^2 a_{max}(T^2 - T) \\
    &= y + \frac{1}{2} \left( \frac{v|v|}{a_{max}} \right) + \frac{1}{2} v t_s .
\end{array}
\end{align*}
Taking the travel direction into consideration using $\sgn(v)$, we can then impose the lane centering constraint in \eqref{eq:LCC} as:
\begin{align*}
\begin{array}{c} 
    \sgn(v)y_T = \sgn(v)y + \frac{v^2}{2 a_{max}} + \frac{1}{2}|v| t_s \leq y_{max}  \\
    \Leftrightarrow v^2 + |v|a_{max} t_s \leq 2 a_{max}(y_{max} - \sgn(v)y) .
    \end{array}
\end{align*}
Completing the square yields
\begin{equation*}
    \left( \hspace{-1pt} |v| \hspace{-1pt} + \hspace{-1pt} \frac{1}{2}a_{max}t_s \hspace{-1pt} \right)^2 \hspace{-4pt}- \frac{1}{4}a_{max}^2 t_s^2 \hspace{-1pt} \leq \hspace{-1pt} 2 a_{max}(y_{max} - \sgn(v)y),
\end{equation*}
and considering its square root leads to our choice of $h_{LK}(x)$ in \eqref{eq:h_lk}. Intuitively, this $h_{LK}(x)\ge 0$ ensures that for any state $x$, there is enough time in the future to come to a complete stop (and switch direction) before reaching the lane boundary. Since the system states are continuous, this includes the case that the lateral displacement at the next time step starting at $y$ with velocity $v$ does not violate the lane centering constraint; thus, the controlled invariance condition in \eqref{eq:DT-CBF} holds and $h_{LK}$ is a DT-CBF for $\mathcal{S}_{LK}$.

Next, we show that $K_{\mathcal{S}_{LK}}$ can be expressed as mixed-integer linear constraints using the composition tools for piecewise functions (as discussed in Remark \ref{rem:MIL}).
Now, for $x_k=\begin{bmatrix}y_k & \nu_k \end{bmatrix}^\top$ and $v_k=\frac{y_{k+1}-y_k}{t_s}$ with $y_{k+1}=y_k+t_s (\nu_k+V_0 \psi_k)$ and the following definition
\begin{equation*}
    \eta_k^{\pm} \triangleq \sqrt{2 a_{max} (y_{max} \mp y_{k+1}) + \frac{1}{4}a_{max}^2 t_s^2} - \frac{1}{2}a_{max} t_s,
\end{equation*}
the controlled invariance condition $h_{LK}(x_{k+1}) \geq 0$ can be written as a piecewise condition
\begin{equation}
\begin{array}{c}
  \displaystyle  \eta_k^+ + \frac{y_{k+1}}{t_s} - \frac{y_{k+2}}{t_s} , \text{ if } v_{k+1} \geq 0 , \\
  \displaystyle  \eta_k^- - \frac{y_{k+1}}{t_s} + \frac{y_{k+2}}{t_s} , \text{ if } v_{k+1} < 0 ,
    \label{eq:PW}
\end{array}
\end{equation}
where $y_{k+2} = y_{k+1} + t_s z_k + t_s^2 \frac{C_f}{M}u_k$ with $z_k \triangleq V_0\psi_{k+1} + (1 + t_s \alpha)\nu_k + t_s \beta r_k$, $\alpha = -\frac{C_f + C_r}{MV_0}$, $\beta = \frac{bC_r - aC_f}{MV_0} - V_0$ and $\psi_{k+1}=\psi_k+t_s (r_k-r_{d,k})$.
Then, using the composition tool for piecewise functions in \eqref{PLBF}, the piecewise condition in (\ref{eq:PW}) can be rewritten as follows:
\begin{equation}
\tag{LC-CBF}
\begin{array}{rl}
    \eta_k^+ - z_k - t_s \frac{C_f}{M} u_k + s_1 \geq 0 ,& \
    z_k + t_s \frac{C_f}{M} u_k + s_1 \geq 0, \\
    \eta_k^- + z_k + t_s \frac{C_f}{M} u_k + s_2 \geq 0 , &\
    -z_k - t_s \frac{C_f}{M} u_k + s_2 > 0 , \\
    \text{SOS-1}:\left\{ s_1,s_2 \right\} ,& \ s_1,s_2 \geq 0,
    \label{eq:LC-CBF}
\end{array}
\end{equation}
which are mixed-integer linear constraints on $u_k$.
\end{proof}

It is noteworthy that in the limit when the sampling time $t_s$ tends to zero, our $h_{LK}(x)$ in \eqref{eq:h_lk} becomes the continuous-time CBF in \cite[Eq. (53)]{ames2016control}.

Next, we adopt the optimal control framework with a quadratic cost in \cite{ames2016control} to select the optimal input from the (safe) input set $K_{\mathcal{S}_{LK}}$, as follows:\\[-0.25cm]

\noindent{\textbf{Mixed-Integer Quadratic Program for LK:}} The DT-CBF is combined with a linear feedback controller $u_k = -K(x_k - x_{ff,k})$, where $K$ is a (legacy) controller gain and $x_{ff,k} = \begin{bmatrix} 0 & 0 & 0 & r_{d,k} \end{bmatrix}^\top$, as well as the acceleration and lane centering constraints, \eqref{eq:acc_c} and \eqref{eq:LCC}, respectively, resulting in the following 
mixed-integer quadratic program: 
\begin{equation}
\begin{aligned}
&\mathbf{u_k^*} = && \argmin_{\mathbf{u_k}=[u_k,\delta]^\top} \quad \frac{1}{2} \mathbf{u_k}^\top H \mathbf{u_k} + F^\top \mathbf{u_k} \\
& \text{s.t.} && (\text{\ref{eq:AC}}) \ \text{and} \ (\text{\ref{eq:LC-CBF}}) \ \text{hold}, \\
& && u_k = -K(x_k - x_{ff,k}) + \delta , 
\label{eq:MIQP}
\end{aligned}
\end{equation}
where $\delta$ is a relaxation variable such that the linear feedback controller forms a soft constraint that is only achieved if the required (safety) constraints are not violated, $H \in \mathbb{R}^{2 \times 2}$ is positive definite, and $F \in \mathbb{R}^{2}$.

\subsection{Obstacle Avoidance Setup}

Next, we consider the Obstacle Avoidance (OA) problem as an extension to the LK problem, where in the event that there is an obstacle in the road lane, the vehicle avoiding the obstacle to the left or right can be modeled by a LK problem in which the lane splits into two lanes going around the obstacle on either side, one with a curve rate of $r_{d_1,k}$ and another with a curve rate of $r_{d_2,k}$. 
Obviously, the vehicle cannot remain in both lanes as they split around the obstacle and we encode the choice between the left and right lanes using a conjunction (`OR' or $\vee$) of two barrier functions for each lane with $h_{LK,l}$ and $h_{LK,r}$, i.e., with $(h_{LK,l}\ge 0) \vee (h_{LK,r} \ge 0)$.\\[-0.25cm]

\noindent{\textbf{Mixed-Integer Quadratic Program for OA:}} When incorporated into an optimal control framework as in \eqref{eq:MIQP}, we obtain another mixed-integer quadratic program by virtue of the composition tools we developed in Section \ref{sec:compose}: 
\begin{equation}
\begin{aligned}
&\mathbf{u_k^*} = && \argmin_{\mathbf{u_k}=[u_k,\delta]^\top} \quad \frac{1}{2} \mathbf{u_k}^\top H \mathbf{u_k} + F^\top \mathbf{u_k} \\
& \text{s.t.} && ((\text{\ref{eq:AC}}_l) \wedge (\text{\ref{eq:LC-CBF}}_l)) \vee ((\text{\ref{eq:AC}}_r) \wedge (\text{\ref{eq:LC-CBF}}_r)) , \\
& && u_k = -K(x_k - x_{ff,k}) + \delta ,\\[-0.75cm] 
\label{eq:MIQP2}
\end{aligned}
\end{equation}
where $(\text{\ref{eq:AC}}_l)$ and $(\text{\ref{eq:LC-CBF}}_l)$ are $(\text{\ref{eq:AC}})$ and $(\text{\ref{eq:LC-CBF}})$ based on $r_{d_1,k}$, while $(\text{\ref{eq:AC}}_r)$ and $(\text{\ref{eq:LC-CBF}}_r)$ are based on $r_{d_2,k}$.

\subsection{Simulation Results}

Table \ref{tab:sim_vals} shows the values of the parameters used in the simulations of both the LK and OA problems.

\begin{table}[h!]
\centering
\caption{Parameter Values Used in Simulations}
\begin{tabular}{|c|c|c|c|}
    \hline
    $V_0$ & 8.33 m/s & $C_f$ & 133000 N/rad \\ \hline
    $C_r$ & 98800 N/rad & M & 1650 kg \\ \hline
    a & 1.11 m & b & 1.59 m \\ \hline
    $I_z$ & 2315.3 $m^2 kg$ & g & 9.81 $m/s^2$ \\ \hline
    $a_{max}$ & $0.3 \times g$ $m/s^2$ & $t_s$ & 0.01 s \\ \hline
\end{tabular}
\label{tab:sim_vals}
\end{table}\vspace{0.15cm}

The feedback gain $K$ was determined using MATLAB's place command to place the poles at $\{0.95,0.8,0.85,0.9\}$. 

\subsubsection{LK Problem} First, we demonstrate the effectiveness of our DT-CBF approach for the LK problem and compare it with the continuous-time approach in \cite{ames2016control}. As shown in Figure \ref{figureP}, with the initial state set to $x_0 = \begin{bmatrix} 
0.5 & 0 & 0 & 0 \end{bmatrix}^\top$, the lateral displacement and acceleration for both DT-CBF and CT-CBF stay within the desired bounds of $\pm 0.9 m$ and $\pm 0.3g$, respectively, but their behaviors are rather different. The lateral acceleration with the DT-CBF is less smooth presumably because of the non-smooth piecewise barrier function, but the lateral displacement remains much closer to zero for the duration of the simulation, meaning the vehicle stays closer to the center of the lane, as desired. On the other hand, the vehicle drifts up to approximately $0.4$ meters from the center of the lane once the road starts to curve at $t = 10$ seconds with the CT-CBF. Moreover, 
since the control input is proportional to the lateral acceleration, it seems to suggest that smaller inputs are needed in the long run when using the DT-CBF.

\begin{figure}[t] 
\centering
\includegraphics[width=0.2495\textwidth,trim=5mm 0mm 4mm 0mm,clip]{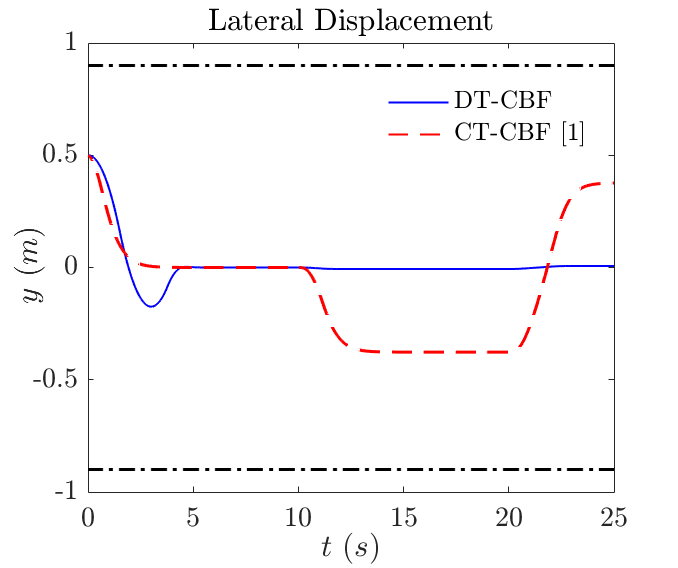}~\hspace{-0.1cm}
\includegraphics[width=0.2495\textwidth,trim=4mm 0mm 4mm 0mm,clip]{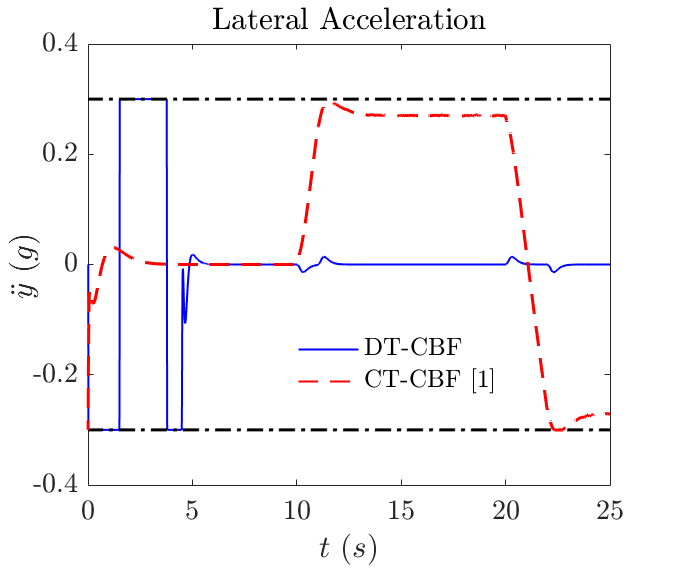}
\caption{Comparison between the proposed DT-CBF approach (blue solid lines) and the CT-CBF in \cite{ames2016control} (red dashed lines).}%
\label{figureP}\vspace{-0.2cm}
\end{figure}

\subsubsection{OA Problem}
An example scenario for the obstacle avoidance problem 
is while driving down a road and noticing an obstacle up ahead where the vehicle either needs to go around the obstacle to the left, or to the right.  As opposed to a vehicle following a curved road and staying within a safe distance of the road center, the road is simulated to curve in two opposite directions $r_{d_1,k}$ and $r_{d_2,k} = -r_{d_1,k}$ 
and the vehicle can choose whether to avoid the obstacle by driving around it to the left or the right (cf. Figure \ref{fig:right}).

\begin{figure}[t]
    \centering \vspace{0.15cm}
    \includegraphics[width = 0.415\textwidth]{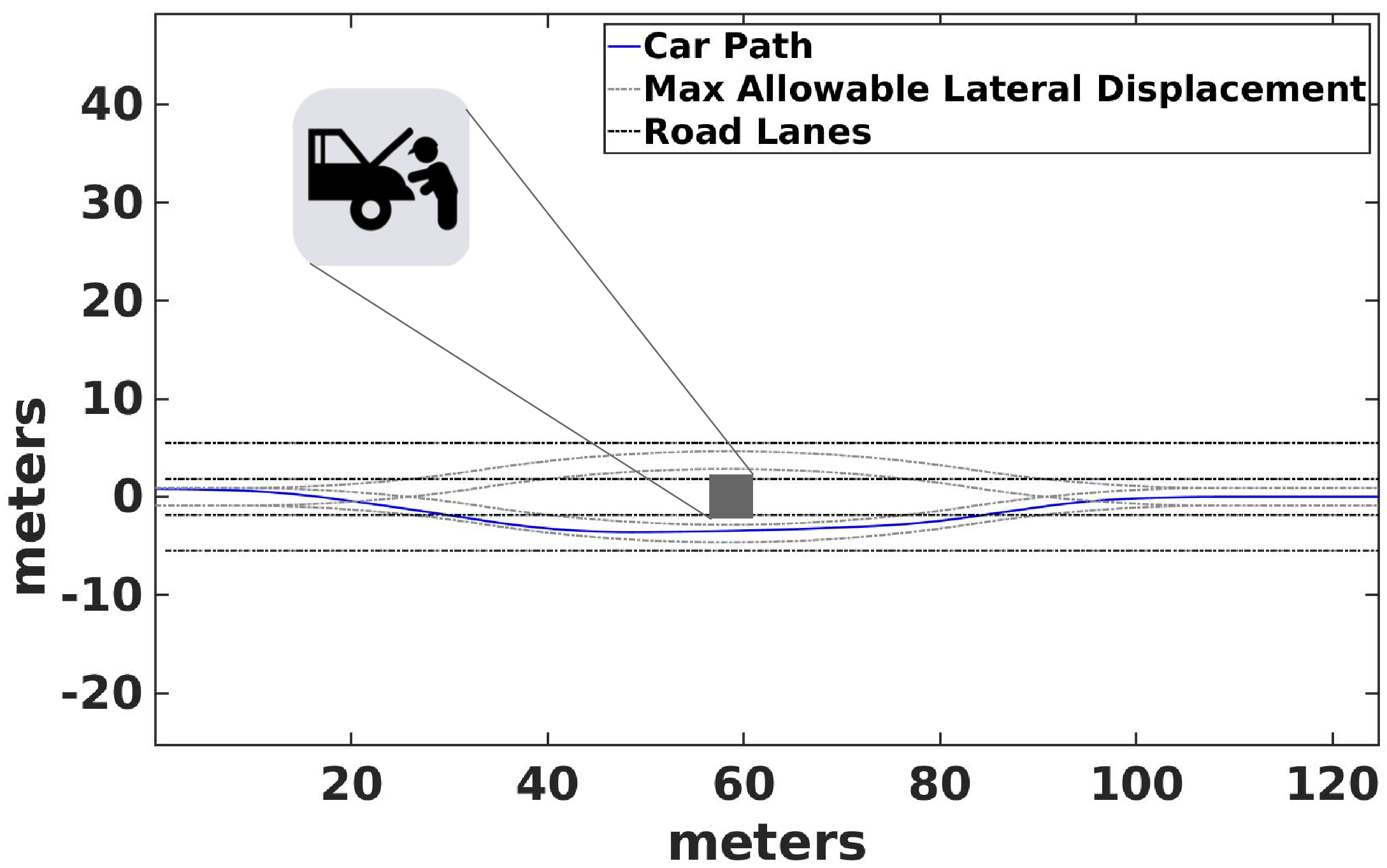}
    \caption{The vehicle must go around the obstacle to the left or to the right. This is represented by two new lanes to follow. The vehicle in this situation 
    chooses to follow the right lane. 
    }
    \label{fig:right}
\end{figure}

\begin{figure*}[!t]
    \centering
    \includegraphics[width = 0.333\textwidth, trim=5mm 0mm 5mm 0mm,clip]{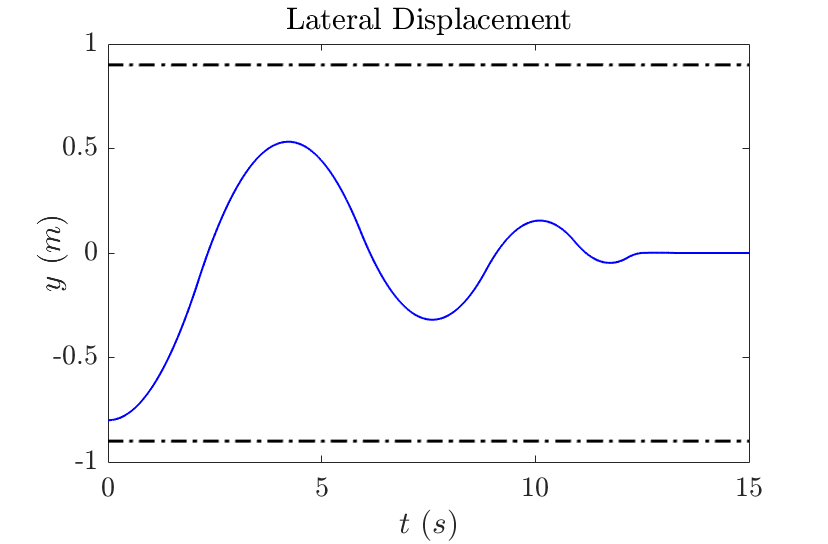}\includegraphics[width = 0.333\textwidth, trim=5mm 0mm 5mm 0mm,clip]{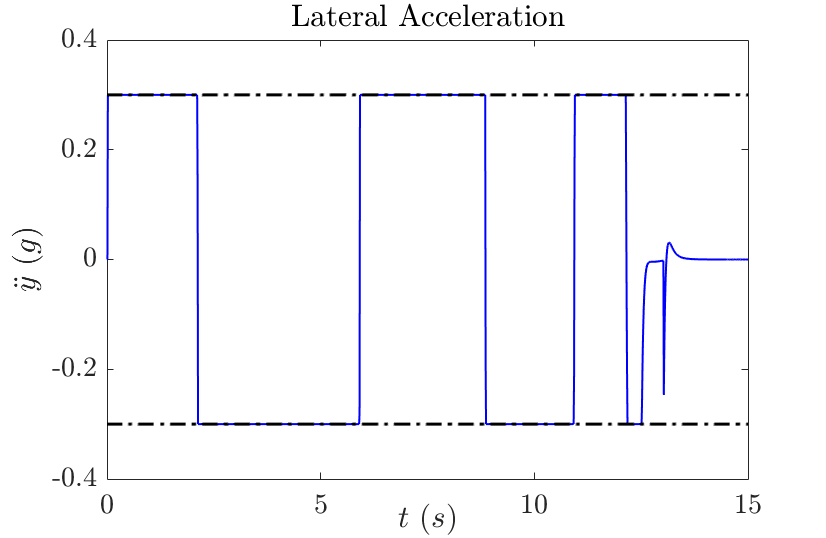}\includegraphics[width = 0.333\textwidth, trim=5mm 0mm 5mm 0mm,clip]{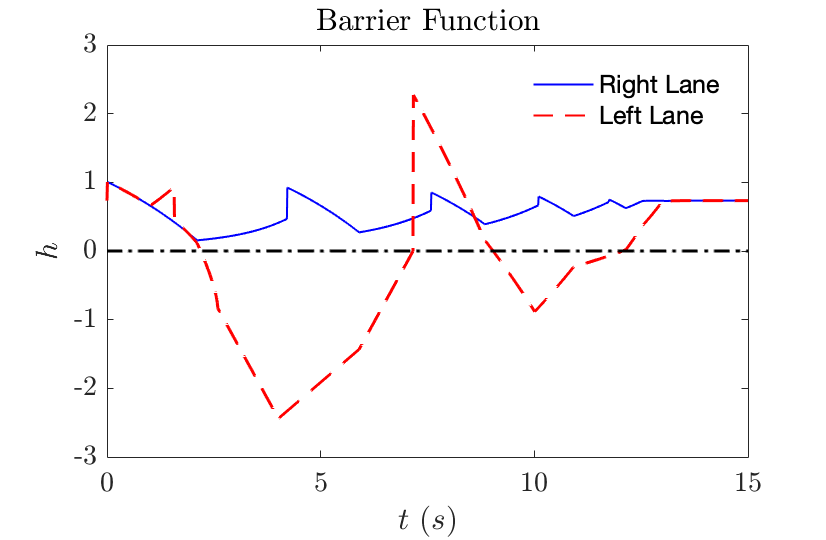}
    \caption{The lateral displacement (left) of the car is bounded by 0.9 meters from the chosen lane. The lateral acceleration (middle) is bounded by $0.3g$. The barrier function (right) that was chosen 
    in the conjunction (`OR') condition stays positive throughout the simulation (blue solid line), while the other does not (red dashed line).}
    \label{fig:3plots_right}\vspace{-0.45cm}
\end{figure*}

To simulate this OA problem, we implemented the mixed-integer quadratic program in 
\eqref{eq:MIQP2} with the initial condition set to $x_0 = \begin{bmatrix} 
-0.8 & 0 & 0 & 0 \end{bmatrix}^\top$, and the results are shown in Figures \ref{fig:right} and \ref{fig:3plots_right}, where
the lateral displacement and lateral velocity remained within the desired constraints, as expected. Moreover, for the chosen lane (to the right in this case), 
the control barrier function $h$ for that lane (cf. Figure \ref{fig:3plots_right}, right, blue solid line) remained positive, but that was not true for the other barrier function (red dashed line). 

From running several simulations, it appears that the vehicle decides to continue accelerating in whatever lateral direction it is already accelerating in. The lane split in all simulations occurred at $t = 1$ second. In Fig \ref{fig:3plots_right} at $t = 1$ second the car has a lateral acceleration of approximately $0.3g$ which indicates accelerating to the right. Therefore, the car chooses to follow the right lane around the obstacle. Conversely, if the car had a negative lateral acceleration, e.g., with $y_0=0.8$ m, it would choose to follow the left lane. 



\section{Conclusion} \label{sec:conclude}

This paper presented a novel formulation for control barrier functions for ensuring the safety of discrete-time systems. This formulation was shown to be necessary and sufficient for controlled invariance and less restrictive than existing formulations. 
In addition, we proposed nonlinear discrete-time control barrier functions for partially control affine systems, whose controlled invariance conditions are affinely affected by the control input, which meant that they can be included as tractable constraints in safety optimal control problems for a broader range of applications and safety conditions than the state-of-the-art. Furthermore, we derived mixed-integer formulations for Boolean compositions of multiple control barrier functions as well as for piecewise control barrier functions. 
Finally, these new sets of 
discrete-time control barrier function tools were applied and tested in simulation for lane keeping and obstacle avoidance. 

\bibliographystyle{unsrt}
\bibliography{biblio}
\end{document}